\documentclass
[final,1p,times]{elsarticle}
\journal{E}
\usepackage{tikz-cd}
\usepackage{extpfeil}

\usepackage{amssymb}
\usepackage{amsthm}
\usepackage{latexsym}
\usepackage{amsmath}
\usepackage{color}
\usepackage{graphicx}
\usepackage{indentfirst}
\usepackage{mathrsfs}
\usepackage{lipsum}
\usepackage{comment}
\usepackage{physics}
\usepackage{subcaption} 
\allowdisplaybreaks

\newtheorem{theorem}{\color{black}\indent \textbf{Theorem}}[section]
\newtheorem{lemma}{\color{black}\indent Lemma}[section]
\newtheorem{proposition}{\color{black}\indent Proposition}[section]
\newtheorem{definition}{\color{black}\indent Definition}[section]
\newtheorem{remark}{\color{black}\indent Remark}[section]
\newtheorem{corollary}{\color{black}\indent Corollary}[section]
\newtheorem{example}{\color{black}\indent Example}[section]

\begin{document}
\title{Canonical and Canonoid transformations for Hamiltonian systems on locally conformal symplectic manifolds}
\author{
Rafael Azuaje$^{a}$,
Xuefeng Zhao$^{b}$ \\[1ex]
$^{a}$ Departamento de física, Universidad Autónoma Metropolitana unidad Iztapalapa, 09340 Mexico City, Mexico.\\
$^{b}$College of Mathematics, Jilin University, Changchun, 130012, P. R. China \\[1ex]
\texttt{razuaje@xanum.uam.mx}, \texttt{zhaoxuef@jlu.edu.cn}
}
\begin{abstract}
This paper is focused on the development of the notions of canonical and canonoid transformations within the framework of Hamiltonian Mechanics on locally conformal symplectic manifolds. Both, time-independent and time-dependent dynamics are considered. Noether-like theorems relating one-parameter groups of transformations with canonical and noncanonical symmetries, are formulated, proved as well as illustrated with elementary examples.
\end{abstract}
\date{}
\thanks{*Corresponding author at: razuaje@xanum.uam.mx}
\maketitle

\section{Introduction}

The notion of canonical transformation is well established in Classical Mechanics. It is well known for the research community in the subject, in fact it is presented in classical textbooks \cite{Landau,Calkin,Goldstein} as well as in research papers where it is studied from a geometric modern approach \cite{Asorey1983,cariñena1985canonical,azuaje2025canonical}. On the other hand, the so called canonoid transformations have shown to be as relevant as the canonical ones since they are related to bi-Hamiltonian structures. The concept of canonoid transformation was introduced in \cite{Saletan} and previously studied from a geometric approach in \cite{CR88,cariñena1989,Carinena}. Canonoid transformations have also been studied within the framework of Hamiltonian dynamics on Poisson and contact manifolds \cite{Rastelli2015,azuajecanonical2023}.

The notion of LCS (locally conformal symplectic) manifolds was initially introduced by Lee \cite{Lee}, and subsequently examined by a number of researchers, including Libermann \cite{Libermann}, Lefebvre \cite{Lefebvre}, Gray--Hervella \cite{Gray}, Vaisman \cite{Vaisman}, Eliashberg--Murphy \cite{Eliashberg}, and Chantraine--Murphy \cite{Chantraine}. These manifolds are of particular importance in the study of certain mechanical systems, such as Gaussian isokinetic dynamics and Nosé--Hoover dynamics, as explored by Wojtkowski \cite{Wojtkowski}.

In addition, LCS geometry has deep ties to other mathematical areas, especially in relation to compact complex surfaces. For example, several non-Kähler compact complex surfaces have been shown to admit locally conformal Kähler structures, as demonstrated in the works of Belgun \cite{Belgun}, Verbitsky \cite{Verbitsky}, and others. Moreover, LCS manifolds, in conjunction with contact manifolds \cite{Borman,Bourgeois,Bowden,Pardon}, can be seen as specific instances of transitive Jacobi manifolds \cite{Guedira,Marle} in both even and odd dimensions. Specifically, a LCS manifold of the first kind serves as a particular case of a contact pair, implying the presence of an intrinsic transversely symplectic foliation \cite{Bande}. This characterizes the geometry of such manifolds as a distinct type of transversely symplectic geometry.

Hamiltonian Mechanics on LCS manifolds studies certain systems whose dynamics locally can be modeled by a Hamiltonian system on a symplectic manifold, but it can not be done globally. As stated in \cite{Esen}, some systems with nonlocal potentials, and systems defined by parts, can be described as Hamiltonian systems on LCS manifolds (see \cite{Esen} for a concrete example of a Hamiltonian system on a LCS manifold).

The aim of this paper is to study the notions of canonical and canonoid transformations within the framework of Hamiltonian Mechanics on LCS manifolds. In order, concepts of canonical and canonoid transformations, as extensions of the established notions to the case of Hamiltonian systems on LCS manifolds, are proposed; in addition, it is shown that there exists canonical and noncanonical symmetries related to such transformations, specifically, it is shown that some canonical symmetries with a well defined property, are infinitesimal generators of one-parameter groups of canonical transformations, and the noncanonical symmetries called scaling symmetries, are shown to be infinitesimal generators of one-parameter groups of canonoid transformations. Infinitesimal generator of one-parameter groups of canonoid transformations are completely characterized. It is worth remarking that the development of the notions and ideas include time-independent and time-dependent dynamics. A remarkable result is the establishment of a special class of the canonical symmetries referred to as Noether symmetries, which are related to the so called dissipated quantities, as infinitesimal generators of one-parameter groups of canonical transformations that leave the Hamiltonian invariant (invariance transformations).

The paper is organized as follows. In Section 2, we recall the concept of LCS  manifolds, the time-independent Hamiltonian systems defined on them, and introduce (local) time-dependent Hamiltonian systems. In Section 3, we revisit the canonical transformations for time-independent Hamiltonian systems, define the canonical transformations for time-dependent Hamiltonian systems, and establish the relationship between the local time-dependent Hamiltonian systems and the infinitesimal generators of a family of canonical transformations. In Section 4, we present the Canonoid transformations for both time-independent and time-dependent Hamiltonian systems. We also prove several properties of the scaling symmetries generated by the infinitesimal generators of a family of Canonoid transformations. Furthermore, we introduce a more general class of scaling symmetries and discuss their properties. In Section 5, we define the dissipation quantities for both time-independent and time-dependent Hamiltonian systems. We also define Noether symmetries and Cartan symmetries, and show how to use these symmetries to identify conserved quantities and dissipation quantities.

\section{Hamiltonian mechanics on locally conformal symplectic manifolds}

First we recall the definition of locally conformal symplectic manifolds. Let \( M \) be a smooth manifold. An \emph{almost symplectic form} is a non-degenerate 2-form on \( M \). An almost symplectic form \( \Omega \) is said to be \emph{locally conformal symplectic} if there exists an open covering \(\{U_\lambda\}_{\lambda\in \Lambda}\) of \( M \) and a family of smooth functions \(\{f_\lambda: U_\lambda \to \mathbb{R}\}\) such that \(e^{-f_\lambda} \cdot (\Omega|_{U_\lambda})\) is symplectic for each \(\lambda\).  
	
The non-degeneracy of \( \Omega \) implies that \( df_\lambda = df_\mu \) on \( U_\lambda \cap U_\mu \). Consequently, the forms \(\{df_\lambda\}\) glue together to define a globally defined closed 1-form \(\theta\), which satisfies
	\[
	d\Omega = \theta \wedge \Omega.
	\]
If \(\dim M \geq 4\), the form \(\theta\) is uniquely determined by \(\Omega\) (see \cite{Vaisman}).  
	
Throughout this paper, we denote by \(\mathfrak{X}(M)\) the space of smooth vector fields on \( M \). For \( k \in \mathbb{N} \), we write \(\Omega^k(M)\) for the space of smooth \( k \)-forms on \( M \); in particular, \(\Omega^0(M) = C^\infty(M)\) is the space of smooth functions on \( M \), and we set
	\[
	\Omega^*(M) = \bigcup_{k \in \mathbb{N}} \Omega^k(M).
	\]
	
\begin{definition}
Let \( M \) be a \(2n\)-dimensional differentiable manifold and let \( \Omega \) be a non-degenerate 2-form on \( M \).  
If there exists a closed 1-form \(\theta\) such that
		\[
		d\Omega = \theta \wedge \Omega,
		\]
then the triple \((M,\Omega,\theta)\) is called a \emph{locally conformal symplectic} (LCS) manifold, \(\Omega\) is called a LCS form with Lee form \(\theta\), and the pair \((\Omega,\theta)\) is called a LCS structure.  
		
In particular, if \(\theta\) is exact, then \((M,\Omega,\theta)\) is called a \emph{globally conformal symplectic} (GCS) manifold, \(\Omega\) is called a GCS form with Lee form \(\theta\).
\end{definition}

\subsection{Time-independent dynamics}
We recall the formalism of time-independent Hamiltonian dynamics on LCS manifolds \cite{Vaisman,Esen}.

As described in \cite{Esen}, for any non-degenerate 2-form $\Omega,$ we can define a musical isomorphism
	\begin{equation}\label{Musical}
		\Omega^\flat : \mathfrak{X}(M) \to \Omega^1(M) : X \mapsto X\lrcorner\omega,
	\end{equation}
where $X\lrcorner$ is the interior product. We denote the inverse of isomorphism \eqref{Musical} by $\Omega^\sharp$. When pointwise evaluated, the musical mappings $\Omega^\flat$ and $\Omega^\sharp$ induce isomorphisms from $TM$ to $T^*M$ and from $T^*M$ to $TM$, respectively. 
	\begin{definition}
		A vector field  $X$ on a LCS manifold $(M,\Omega,\theta)$ is a Hamiltonian vector field with Hamiltonian function $f$ if it satisfies 
		$$\Omega^\flat(X)=X\lrcorner\Omega=df-f\theta=:d^\theta f,$$
		where $d^\theta=d-\theta\wedge,$ which is called a Lichnerowicz–deRham differential \cite{Esen}. 

        When the 1-form \( V\lrcorner \Omega \) is closed but not exact with respect to $d^{\theta}$, i.e., the function $f$ is locally defined, the vector field $X$ is called a locally Hamiltonian vector field.
	\end{definition}

\begin{remark}
Time-independent Hamiltonian dynamics on a LCS manifold is defined by a Hamiltonian vector field. Let $(M,\Omega,\theta)$ be a LCS manifold and $H\in C^{\infty}(M)$, the mechanical system defined by $X_{H}$ on $M$ is a Hamiltonian system denoted by $(M,\Omega,\theta,H)$.
\end{remark}

\begin{example}
Consider the manifold
\[
M = T^*Q,
\]
with local coordinates \((q_i,p_i)\), \(i=1,\dots,n\).  
The canonical Liouville 1-form is
\[
\lambda = p_i\,dq_i,\qquad d\lambda = dp_i \wedge dq_i.
\]

Let \(\theta = \theta_i(q)\,dq_i\) be a closed 1-form on the base manifold \(Q\) (pulled back to \(T^*Q\)).  
We can get the LCS   form 
\[
\Omega = d_\theta\lambda
= d\lambda - \theta \wedge \lambda
= dp_i \wedge dq_i - p_i\,\theta \wedge dq_i,
\]
where the term \(\theta\wedge\lambda = \theta\wedge(p_i\,dq_i)\) has been expanded for later use.

Let \(H=H(q,p)\) be a smooth Hamiltonian. The Hamiltonian vector field \(X_H\) is defined by
\[
X_H\lrcorner\,\Omega = d^\theta H
= dH - H\,\theta,
\]
where the twisted differential on 0-forms acts as
\(d^\theta H = dH - H\,\theta.\)
Assume
\[
X_H = A_i(q,p)\,\frac{\partial}{\partial q_i}
+ B_i(q,p)\,\frac{\partial}{\partial p_i}.
\]
We can compute
\[
X_H\lrcorner d\lambda
= X_H\lrcorner(dp_i \wedge dq_i)
= B_i\,dq_i - A_i\,dp_i.
\]
Using \(X\lrcorner(\alpha\wedge\beta)
= (X\lrcorner\alpha)\,\beta - \alpha\wedge(X\lrcorner\beta)\), we find
\[
X_H\lrcorner\theta = \theta_j\,A_j,\qquad
X_H\lrcorner\lambda = p_i\,A_i,
\]
and hence
\[
X_H\lrcorner(\theta\wedge\lambda)
= \theta(X_H)\,\lambda - \theta\wedge(p_i\,A_i)
= \theta(X_H)\,p_k\,dq_k - p_i A_i\,\theta.
\]
Combining the above,
\[
X_H\lrcorner\Omega
= X_H\lrcorner d\lambda - X_H\lrcorner(\theta\wedge\lambda)
= \bigl(B_k - \theta(X_H)\,p_k + p_i A_i\,\theta_k\bigr)\,dq_k
- A_k\,dp_k,
\]
where \(\theta = \theta_k\,dq_k.\)
We have
\[
d^\theta H
= dH - H\,\theta
= \frac{\partial H}{\partial q_k}\,dq_k
+ \frac{\partial H}{\partial p_k}\,dp_k
- H\,\theta_k\,dq_k.
\]
Equating the coefficients of \(dp_k\) and \(dq_k\) in
\(X_H\lrcorner\Omega = d^\theta H,\)
we get
\[
-\,A_k = \frac{\partial H}{\partial p_k}
\quad\Longrightarrow\quad
A_k = -\,\frac{\partial H}{\partial p_k},
\]
and
\[
B_k - \theta(X_H)\,p_k + p_i A_i\,\theta_k
= \frac{\partial H}{\partial q_k} - H\,\theta_k.
\]
Substituting \(\theta(X_H) = \theta_j\,A_j\), we obtain
\[
B_k
= \frac{\partial H}{\partial q_k}
- H\,\theta_k
+ (\theta_j A_j)\,p_k
- (p_i A_i)\,\theta_k.
\]
Finally, using \(A_i = -\,\dfrac{\partial H}{\partial p_i}\),
we arrive at the explicit LCS Hamilton equations:
\begin{align*}
	\begin{cases}
		\text{}	\dot{q}_k &= A_k
		= -\,\frac{\partial H}{\partial p_k},\\
		\text{}\dot{p}_k &= B_k
		= \frac{\partial H}{\partial q_k}
		- H\,\theta_k
		- p_k\,\theta_j\,\frac{\partial H}{\partial p_j}
		+ \theta_k\,p_j\,\frac{\partial H}{\partial p_j}
	\end{cases}
\end{align*}
Here we have used
\[
\theta_j A_j = -\,\theta_j\,\frac{\partial H}{\partial p_j},\qquad
p_i A_i = -\,p_i\,\frac{\partial H}{\partial p_i}.
\]
\end{example}

\subsection{Time-dependent dynamics}
First, we briefly present the fundamentals of time-dependent mechanics following \cite{Abraham2}. 
Let \( Q \) be a smooth manifold. We denote by \( TQ \) and \( T^*Q \) the tangent and cotangent bundles of \( Q \), respectively. 

A map
\[
X : \mathbb{R} \times Q \to TQ, \qquad (t,q) \longmapsto X(t,q)
\]
is called a \emph{time-dependent vector field} if, for each \( t \in \mathbb{R} \), the map
\[
X_t : Q \to TQ, \qquad q \longmapsto X(t,q)
\]
is a vector field on \( Q \). Hence, any time-dependent vector field on a manifold can be regarded as a family of time-independent vector fields \(\{X_t\}_{t \in \mathbb{R}}\) smoothly depending on the parameter \( t \in \mathbb{R} \). The integral curves of $X$ are functions $\varphi:I\subset\mathbb{R}\longrightarrow Q$ such that $\dot{\varphi}(t)=X(t,\varphi(t))$.

Each time-dependent vector field is uniquely associated with a (time-independent) vector field on \( \mathbb{R} \times Q \) defined by
\[
\tilde{X} : \mathbb{R} \times Q \to T(\mathbb{R} \times Q) \simeq T\mathbb{R} \times TQ, 
\qquad (t,q) \longmapsto \bigl((t,1),\, (q, X(t,q))\bigr),
\]
so that
\[
\tilde{X} = \frac{\partial}{\partial t} + X(t,q).
\]
We refer to \(\tilde{X}\) as the \emph{suspension} or \emph{autonomization} of the time-dependent vector field \( X \). The integral curves of $\tilde{X}$ are of the form $\psi(t)=(t,\varphi(t))\in \mathbb{R}\times Q$ with $\varphi(t)\in Q$ and $\varphi$ integral curves of $X$.

Now we consider time-dependent Hamiltonian dynamics on LCS manifolds (see \cite{Chinea,ZSR2023} for different approaches to time-dependent dynamics on LCS manifolds). We consider a LCS manifold $(M,\Omega,\theta)$ and a time-dependent smooth function on $M$, i.e., a function $H\in C^{\infty}(\mathbb{R}\times M)$; following \cite{ZSR2023}, let \( H_t(m) := H(t,m) \), by fixing $t$, we can get the dynamics 
\begin{align}\label{TH2}
X_{H_t}\lrcorner \Omega = d^\theta H_t
\end{align}
on $M$, where \( X_{H_t}\) is a Hamiltonian vector field on \( M \) associated with the Hamiltonian \( H_t \).  The family of vector fields \(\{X_{H_t}\}\) defines a time-dependent vector field
\[ 
X_H : \mathbb{R} \times M \rightarrow TM, \quad (t,m) \longmapsto X_{H_t}(m).
\]
So, the natural phase space for the time-dependent Hamiltonian dynamics determined by $H$ is $\mathbb{R}\times M$ (called the extended phase space \cite{Struckmeier2005,azuaje2025canonical}). We are interested in the dynamics on the extended phase space $\mathbb{R}\times M$. Let \(\hat{\Omega} = \mathrm{pr}^*\Omega\) be the pullback two-form on \(\mathbb{R} \times M\), with \(\hat{\theta} = \mathrm{pr}^*\theta\) being the corresponding one-form on \(\mathbb{R} \times M\), where $ \mathrm{pr}$ is the projection: $\mathrm{pr}: \mathbb{R} \times M \longrightarrow M.$ The vector field $X_{H}$ on $\mathbb{R}\times M$ satisfies $X_{H}\lrcorner dt=0$ and
\begin{equation}
X_{H}\lrcorner \hat{\Omega}=d^{\hat{\theta}}H-\frac{\partial H}{\partial t}dt.
\end{equation}
We state the following definition.
\begin{definition}
A vector field \( X \) on \(  \mathbb{R} \times M \) is a Hamiltonian vector field if 
\begin{equation}
X\lrcorner dt=0 \quad\textit{and}\quad X\lrcorner \hat{\Omega}=d^{\hat{\theta}}f-\frac{\partial f}{\partial t}dt,
\end{equation}
for some function $f\in C^{\infty}(\mathbb{R}\times M)$. $X$ is called the Hamiltonian vector field for $f$ and it is denoted by $X_{f}$. If $f$ is locally defined then $X$ is called a locally Hamiltonian vector field on $\mathbb{R}\times M$.
\end{definition}
\begin{remark}
The Hamiltonian dynamics on $(\mathbb{R}\times M)$ with time-dependent Hamiltonian $H$, is given by the autonomization $\tilde{X}_{H}=X_{H}+\frac{\partial}{\partial t}$ of $X_{H}$, and we call $(\mathbb{R}\times M,\hat{\Omega},\hat{\theta},H)$ a time-dependent Hamiltonian system on the LCS manifold $(M,\Omega,\theta)$.
\end{remark}

\section{Canonical transformations}

In Classical Hamiltonian Mechanics from a geometric approach, coordinate transformations are represented by diffeomorphisms on phase space; those that preserve the Hamiltonian structure up to a proportionality factor, regardless of the specific form of the Hamiltonian function, are referred to as non-strictly canonical transformations. These transformations are characterized by the existence of a real number $\lambda$, called the valence, such that the Poisson bracket of two transformed functions is $\lambda$ times the Poisson bracket of the original functions, after they have been transformed \cite{Carinena,Saletan,Saletan2}. The set of non-strictly canonical transformations forms a group, and the subset of strictly canonical transformations, corresponding to $\lambda=1$, forms a normal subgroup.   Within the framework of symplectic geometry, strictly canonical transformations are represented by symplectomorphisms in phase space, i.e, diffeomorphisms that preserve the symplectic form.
	
Next,  we present the concept of (time-independent) strictly canonical transformation in the framework of LCS manifolds as presented in \cite{Zhao1}. Let $(M,\Omega,\theta)$ be a LCS manifold.
\begin{definition}
A strictly canonical transformation for $(M,\Omega,\theta)$ is a diffeomorphism $\Phi$ on $M$ such that $\Phi^*(\Omega)=\Omega$.
\end{definition}
\begin{remark}
 In \cite{Zhao1} it has been shown that, indeed, (time-independent) strictly canonical transformations preserve the LCS structure, i.e., $\Phi$ is a (time-independent) strictly canonical transformation for the LCS manifold $(M,\Omega,\theta)$ if and only if,
\begin{equation}
\Psi^{*}\Omega=\Omega\quad\textit{and}\quad \Psi^{*}\theta=\theta. 
\end{equation}
\end{remark}
\begin{remark}\label{R3}
As shown in \cite{Zhao1}, for a Hamiltonian vector field \( X_f \), we have
\[
L_{X_f} \Omega = \theta(X_f)\, \Omega.
\]
Therefore, \( X_f \) preserves the form \( \Omega \) if and only if \( \theta(X_f) = 0 \). In this case, the flow \( \Phi_s \) generated by \( X_f \) satisfies
\[
\Phi_s^* \Omega = \Omega,
\]
which means that all the maps \( \Phi_s \) are canonical transformations.
\end{remark}

In what follows we shall refer to  strictly canonical transformations just as canonical transformations. We are interested in studying time-dependent canonical transformations within the framework of LCS manifolds.

We now introduce the following definition of a canonical transformation on the extended phase space determined by a LCS manifold.

\begin{definition}
	Let 
	\(\hat{\Omega} = \mathrm{pr}^*\Omega\) be the pullback two-form on \(\mathbb{R} \times M\),
	with \(\hat{\theta} = \mathrm{pr}^*\theta\) being the corresponding one-form on \(\mathbb{R} \times M\), where $ \mathrm{pr}$ is the projection: $ \mathrm{pr}: \mathbb{R} \times M \longrightarrow M$.
	A smooth map
	\[
	F : \mathbb{R} \times M \to \mathbb{R} \times M
	\]
	is called a \emph{(strict) canonical transformation} (on the extended phase space) if the following conditions are satisfied:
	\begin{enumerate}
		\item[(i)] \(F\) is a diffeomorphism;
		\item[(ii)] \(F\) preserves the time coordinate, i.e., \(F^*t = t\), or equivalently, the following diagram is commutative:
		\[
		\begin{tikzcd}
			\mathbb{R} \times M \arrow[r, "F"] \arrow[d, "\mathrm{pr}_{\mathbb{R}}"'] & 
			\mathbb{R} \times M \arrow[d, "\mathrm{pr}_{\mathbb{R}}"] \\
			\mathbb{R} \arrow[r, equals] & \mathbb{R}
		\end{tikzcd}
		\]
		\item[(iii)] \(F\) satisfies \(F^*\hat{\theta} = \hat{\theta}\);
		\item[(iv)] There exists a function \(K_F \in C^\infty(\mathbb{R} \times M)\) such that
		\[
		F^*\hat{\Omega} = \Omega_{K_F}, \qquad 
		\text{where} \quad \Omega_{K_F} := \hat{\Omega} + d^{\hat{\theta}}K_F \wedge dt.
		\]
	\end{enumerate}
    \end{definition}
    \begin{remark}
       We have that if $F$ is a canonical transformation then, for each ﬁxed value of $t$, $F|_t$ satisfies
       $$(F|_t)^*\Omega=\Omega.$$
    \end{remark}

\begin{remark}
Let us remember that a (local) one-parameter group of canonical transformations is a (local) flow \(\Psi\) (a one-parameter group of transformations \(\{\Psi_s\}\)) on the phase space of a Hamiltonian system such that \(\Psi_s\) is a canonical transformation for each \( s \).
\end{remark}

We present the following result which extends the well known result for Hamiltonian systems on symplectic manifolds that states that infinitesimal generators of one-parameter groups of canonical transformations are Hamiltonian vector fields \cite{spivak2010,azuaje2025canonical}.
\begin{theorem}
Let \(\{\Psi_s\}\) be a (possibly local) one-parameter group of transformations on the extended phase space \((M \times \mathbb{R}\) and \(X \in \mathfrak{X}(M \times \mathbb{R})\) its infinitesimal generator which satisfies $X\lrcorner pr^*\theta=0$. \(\{\Psi_s\}\) is a one-parameter group of canonical transformations if and only if \(X\) is a locally time-dependent Hamiltonian vector field.
\end{theorem}
\begin{proof}
First, suppose that \(X\) is a locally time-dependent Hamiltonian vector field. By definition, this means that \(X \lrcorner\, dt = 0\), and in a neighborhood \(U \subset M\) of any point \(p \in M\), there exists a function \(f: U \times \mathbb{R} \to \mathbb{R}\) such that
\[
(X \lrcorner\, \Omega)|_{U \times \mathbb{R}} = d^\theta f - \frac{\partial f}{\partial t} \, dt.
\]
Since \(X \lrcorner\, dt = 0\), for each fixed \(t\), the vector field \(X\) can be regarded as a vector field on \(M\), and satisfies
\[
(X \lrcorner\, \Omega)|_U = d^\theta f|_U, \qquad (X \lrcorner\, \theta)|_U = 0.
\]
This implies that \(X\) is a locally Hamiltonian vector field on \((M, \Omega,\theta)\) satisfying \(X\lrcorner\theta=0\). Therefore, by Remark~\ref{R3}, \(X\) generates a (possibly local) one-parameter group of canonical transformations \(\{\Psi_s\}\) such that
\[
\Psi_s^* \Omega = \Omega.
\]
To extend these transformations to \(M \times \mathbb{R}\), we define \(\Psi_s\) to preserve the temporal parameter \(t\). This yields a one-parameter group of transformations \(\Psi_s\) on \(M \times \mathbb{R}\) such that the two-form \(\Psi_s^* \Omega - \Omega\) is closed and vanishes for each fixed \(t\).

Since \(\Psi_s^* \Omega - \Omega\) is closed, locally there exists a 1-form \(\alpha\) such that
\[
\Psi_s^* \Omega - \Omega = d\alpha.
\]
Moreover, as this difference vanishes when restricted to each time slice \(t = \text{const}\), we must have
\[
\Psi_s^* \Omega - \Omega = d(J \, dt) = dJ \wedge dt,
\]
for some smooth function \(J\). Hence, \(X\) is the infinitesimal generator of a one-parameter group of canonical transformations. The same conclusion also holds globally if \(X\) is a globally time-dependent Hamiltonian vector field.

Conversely, suppose that \(\{\Psi_s\}\) is a one-parameter group of canonical transformations that preserves the temporal parameter \(t\). Then \(X \lrcorner\, dt = 0\), and by Remark~\ref{R3}, it follows that \(X\lrcorner\theta = 0\), and for each fixed \(t\), \(X\) is a Hamiltonian vector field on \((M, \Omega,\theta)\). Therefore, in a neighborhood \(U \subset M\) of any point \(p \in M\), there exists a smooth function \(f_t: U \to \mathbb{R}\) such that
\[
(X \lrcorner\, \Omega)|_U = d^\theta f_t.
\]
Note that
\[
d(X \lrcorner\, \Omega) = L_X \Omega - C \lrcorner\, d\Omega = - X \lrcorner\, (\theta \wedge \Omega) = \theta \wedge (X \lrcorner\, \Omega),
\]
which implies
\[
d(X \lrcorner\, \Omega) - \theta \wedge (X \lrcorner\, \Omega) = d^\theta (X \lrcorner\, \Omega) = 0.
\]
Thus, condition \(d^\theta (X \lrcorner\, \Omega) = 0\) is satisfied.

Now, define a smooth function \(f: U \times \mathbb{R} \to \mathbb{R}\) by \(f(x, t) := f_t(x)\). Then we have
\[
(X \lrcorner\, \Omega)|_{U \times \mathbb{R}} = df - \frac{\partial f}{\partial t} \, dt,
\]
which confirms that \(X\) is a locally time-dependent Hamiltonian vector field.
\end{proof}


\section{Canonoid transformations}

In Classical Mechanics, canonoid transformations are presented (from a geometric perspective) as diffeomorphisms on phase space such that the dynamical Hamiltonian vector field is Hamiltonian for the transformed geometric structure \cite{CR88,Carinena,azuajecanonical2023}. In particular, within the context of symplectic geometry, a canonoid transformations for the Hamiltonian system $(M,\omega,H)$ is a diffeomorphism $\Psi$ on $M$ for which there exists a function $K\in C^{\infty}(M)$ such that 
\begin{equation}
X_{H}\lrcorner\Psi^{*}\omega=dK.
\end{equation}
Now we develop the notion of canonoid transformation for Hamiltonian systems on LCS manifolds. Let $(M,\Omega,\theta)$ be a LCS manifold. Given a diffeomorphism $\Psi$ on $M$ we have that $(\Psi^{*}\Omega,\Psi^{*}\theta)$ defines a LCS structure on $M$, indeed,
\begin{equation}
d\Psi^{*}\Omega=\Psi^{*}(d\Omega)=\Psi^{*}(\theta\wedge\omega)=\Psi^{*}\theta\wedge \Psi^{*}\Omega,
\end{equation}
of course $\Psi^{*}\theta$ is a closed 1-form.

\subsection{The time-independent case}
Let $(M,\Omega,\theta,H)$ be a Hamiltonian system; we propose the following definition.
\begin{definition}
\label{decanonoidLCS}
A canonoid transformation for $(M,\Omega,\theta,H)$ is a diffeomorphism $\Psi$ on $M$ for which there exists a function $K\in C^{\infty}(M)$ such that
\begin{equation}
X_{H}\lrcorner \Psi^{*}\Omega=d^{\Psi^{*}\theta}K
\end{equation}
\end{definition}
Definition \ref{decanonoidLCS} is more general than the one presented in \citep{Zhao1} where it is required that the Lee 1-form $\theta$ is invariant under the transformation.

As for canonical transformations, we can also consider (local) one-parameter groups of canonoid transformations. We have the following result (see \citep{azuaje2024scaling} for analogous results on symplectic and contact manifolds).
\begin{theorem}
\label{theoremcanonoid}
$X\in\mathfrak{X}(M)$ is the infinitesimal generator of a one-parameter group of canonoid transformations for $(M,\Omega,\theta,H)$ if and only if $[X,X_{H}]$ is a Hamiltonian vector field.
\end{theorem}
\begin{proof}
Let $X\in \mathfrak{X}(M)$ and $\varphi$ its flow. $X$ is the infinitesimal generator of a one-parameter group of canonoid transformations for $(M,\Omega,\theta,H)$ if and only if 
\begin{equation}
\label{eqcanonoids}
X_{H}\lrcorner \varphi_{s}^{*}\Omega = d^{\varphi_{s}^{*}\theta}K_{s},
\end{equation}
for a family $\lbrace K_{s}\rbrace$ of functions on $M$. By taking $\frac{d}{ds}|_{s=0}$ in equation \eqref{eqcanonoids} we have that $X$ is the infinitesimal generator of 
\begin{equation}
X_{H}\lrcorner L_{X}\Omega=d^{\theta}K-HL_{X}\theta,
\end{equation}
where $K=\frac{d}{ds}|_{s=0}K_{s}$ and $K_{s}|_{s=0}=H$. 
So $X$ is the infinitesimal generator of a one-parameter group of canonoid transformations for $(M,\Omega,\theta,H)$ if and only if
\begin{equation}
X_{H}\lrcorner L_{X}\Omega=d^{\theta}K-HL_{X}\theta,
\end{equation}
for some function $K\in C^{\infty}(M)$.

Now let us see
\begin{equation}
\begin{split}
[X,X_{H}]\lrcorner\Omega &= L_{X}(X_{H}\lrcorner\Omega)-X_{H}\lrcorner L_{X}\Omega\\
&= L_{X}(d^{\theta}H)-X_{H}\lrcorner L_{X}\Omega\\
&= d^{\theta}L_{X}H-HL_{X}\theta-X_{H}\lrcorner L_{X}\Omega\\
&= d^{\theta}L_{X}H-HL_{X}\theta-d^{\theta}K+HL_{X}\theta\\
&= d^{\theta}(L_{X}H-K).
\end{split}
\end{equation}
We conclude that $X$ is the infinitesimal generator of a one-parameter group of canonoid transformations for $(M,\Omega,\theta,H)$ if and only if $[X,X_{H}]$ is a Hamiltonian vector field with Hamiltonian function $L_{X}H-K$, with $K$ such that $X_{H}\lrcorner L_{X}\Omega=d^{\theta}K-HL_{X}\theta$.
\end{proof}

\subsection{Scaling symmetries}
The so called scaling symmetries have been shown to be infinitesimal generators of one-parameter groups of (noncanonical) canonoid transformations for Hamiltonian systems on symplectic and contact manifolds \citep{azuaje2024scaling}. Scaling symmetries are special nonstandard symmetries that rescale the dynamics by a constant factor, i.e., they are  a very special case of dynamical similarities \citep{bravetti2023scaling}. Recently, it has been studied the reduction of Hamiltonian dynamics under scaling symmetries \citep{bravetti2024kirillov}. We propose the following definition of scaling symmetry for Hamiltonian systems on LCS manifolds.
\begin{definition}
A vector filed $X\in\mathfrak{X}(M)$ shall be called a scaling symmetry of degree $\Lambda\in\mathbb{R}$  for the Hamiltonian system $(M,\Omega,\theta,H)$ when
\begin{enumerate}
\item[i)] $L_{X}\Omega=\Omega$,
\item [ii)] $L_{X}\theta=0$, and
\item [iii)] $L_{X}H=\Lambda H$.
\end{enumerate}
\end{definition}
In fact, we have that scaling symmetries are nonstandard symmetries, indeed, they are neither, Hamiltonian vector fields nor dynamical symmetries. In fact, scaling symmetries rescale the dynamics by a constant factor as stated in the following proposition (analogous to lemma 42 in \cite{azuaje2024scaling}).
\begin{lemma}\label{L4}
If $X$ is a scaling symmetry of degree $\Lambda\in\mathbb{R}$  for the Hamiltonian system $(M,\Omega,\theta,H)$ then $L_{X}X_{H}=(\Lambda-1)X_{H}$.
\end{lemma}
\begin{proof}
Let assume that  $X$ is a scaling symmetry of degree $\Lambda$  for $(M,\Omega,\theta,H)$, we have
\begin{equation}
\begin{split}
(L_{X}X_{H})\lrcorner\Omega=[X,X_{H}]\lrcorner\Omega&=L_{X}(X_{H}\lrcorner \Omega)-X_{H}\lrcorner L_{X}\Omega\\
&= L_{X}(d^{\theta}H)-X_{H}\lrcorner\Omega\\
&=d^{\theta}(L_{X}H)-HL_{X}\theta-d^{\theta}H\\
&= d^{\theta}(\Lambda H)-d^{\theta}H\\
&= d^{\theta}((\Lambda-1)H);
\end{split}
\end{equation}
so $L_{X}X_{H}$ is a Hamiltonian vector field with Hamiltonian function $(\Lambda-1)H$, i.e.,
\begin{equation}
L_{X}X_{H}=X_{(\Lambda-1)H}=(\Lambda-1)X_{H}.
\end{equation}
\end{proof}

From the previous lemma we have that for every scaling symmetry $X$, $[X,X_{H}]$ is a Hamiltonian vector field, so from theorem \ref{theoremcanonoid}, which characterizes infinitesimal generators of one-parameter groups of canonoid transformations as vector fields $X$ such that $[X,X_{H}]$ is a Hamiltonian vector field, where $H$ is the Hamiltonian of the system, we have the following result.
\begin{corollary}
Every scaling symmetry is the infinitesimal generator of a one-parameter group of noncanonical canonoid transformations.
\end{corollary}
\begin{remark}
It is worth remarking that canonoid transformations generated by a scaling symmetry are noncanonical transformations since every scaling symmetry is a nonhamiltonian vector field.
\end{remark}

\begin{example}
Let us consider the LCS manifold $M=\mathbb{R}^{4}-\lbrace(q_1,q_2,p_1,p_2):q_1=0,q_2=0\rbrace$ with LCS structure $(\Omega,\theta)$ locally defined by 
\begin{equation}
\Omega=e^{-lnp_1}(dq_{1}\wedge dp_{1}+dq_2\wedge dp_2),\quad \theta=d(-lnp_1)=-\frac{1}{p_1}dp_1,\quad \forall (q_1,q_2,p_1,p_2):p_1>0;
\end{equation}
indeed,
\begin{equation}
d\Omega=-e^{-lnp_1}\frac{1}{p_1}dp_1\wedge dq_{2}\wedge dp_2=\theta\wedge\Omega.
\end{equation}
Now let us consider the Hamiltonian system on $(M,\Omega,\theta)$ with (natural) Hamiltonian function
\begin{equation}
H(q_1,q_2,p_1,p_2)=\frac{1}{2}(p_1^{2}+p_2^{2})-\frac{1}{q_1}-\frac{1}{q_2}.
\end{equation}
We have that 
\begin{equation}
X=q_1\frac{\partial}{\partial q_1}+q_2\frac{\partial}{\partial q_2}-\frac{p_1}{2}\frac{\partial}{\partial p_1}-\frac{p_2}{2}\frac{\partial}{\partial p_2}
\end{equation}
is a scaling symmetry of $(M,\Omega,\theta,H)$ of degree $-1$. Indeed,
\begin{equation}
\begin{split}
L_{X}\Omega &=d(X\lrcorner\Omega)+X\lrcorner d\Omega\\
&=e^{-lnp_1}\left(dq_1\wedge dp_1+\frac{1}{2}dq_2\wedge dp_2+\frac{p_2}{2p_1}dq_2\wedge dp_1-\frac{q_2}{p_1}dp_1\wedge dp_2\right)\\
&+e^{-lnp_1}\left(\frac{1}{2}dq_2\wedge dp_2-\frac{p_2}{2p_1}dq_2\wedge dp_1+\frac{q_2}{p_1}dp_1\wedge dp_2\right)\\
&=e^{-lnp_1}(dq_{1}\wedge dp_{1}+dq_2\wedge dp_2)\\
&=\Omega,
\end{split}
\end{equation}
\begin{equation}
\begin{split}
L_{X}\theta=L_{X}d({-lnp_1})=-d(L_{X}lnp_1)=-d(-\frac{1}{2})=0,
\end{split}
\end{equation}
and
\begin{equation}
L_{X}H=-\frac{1}{2}p_1^{2}-\frac{1}{2}p_2^{2}\textcolor{red}{+}\frac{1}{q_1}\textcolor{red}{+}\frac{1}{q_2}=-H.
\end{equation}
\end{example}

\begin{remark}
A scaling symmetry need not be unique. In fact, for any two scaling symmetries \( X, X' \) of a Hamiltonian \( H \) satisfying $$1-\theta(X)\neq 0,\;\mathrm{or}\; 1-\theta(X')\neq 0,$$ for simplicity, assume that $1-\theta(X)\neq 0,$ one has
\[
X'\lrcorner \Omega =\frac{1-\theta(X')}{1-\theta(X)} X\lrcorner \Omega +\alpha 
\]
where $\alpha$ is a $d^\theta$-closed 1-form satisfying $$X_H\lrcorner\alpha=\left(\frac{1-\theta(X')}{1-\theta(X)}\Lambda - \Lambda'+\theta(X')-\frac{1-\theta(X')}{1-\theta(X)}\theta(X)\right) H,$$
and \( \Lambda', \Lambda \) are the degrees of \( X', X \), respectively. Conversely, any closed 1-form \( \alpha \) with \( X_H\lrcorner \alpha = (\Lambda' - \Lambda) H \) determines another scaling symmetry \( X' \) of degree \( \frac{1-\theta(X')}{1-\theta(X)}\Lambda -\frac{1-\theta(X')}{1-\theta(X)}\theta(X) -\Lambda' +\Lambda+\theta(X')\) via
\[ X'\lrcorner \Omega =\frac{(1-\theta(X'))X\lrcorner \Omega}{1-\theta (X)}  + \alpha.
\]
\end{remark}
\begin{proof}
Let \( X \) and \( X' \) be two scaling symmetries of \( H \) with degrees \( \Lambda \) and \( \Lambda' \), respectively. Then we have
\[
L_X \Omega = \Omega,\quad L_X\theta =0, \quad L_X H = \Lambda H, \qquad 
L_{X'} \Omega = \Omega, \quad L_{X'}\theta=0, \quad L_{X'} H = \Lambda' H.
\]
Since \( \Omega \) is $d^\theta$-closed, i.e. $d^\theta \Omega=0$, we can write:
\begin{align*}
 d^\theta (X\lrcorner \Omega) &=d(X\lrcorner\Omega)-\theta\wedge (X\lrcorner\Omega)\\
 &= L_X \Omega-X\lrcorner d\Omega -\theta\wedge (X\lrcorner\Omega)\\
 &=L_X\Omega-X\lrcorner(\theta\wedge\Omega)-\theta\wedge (X\lrcorner\Omega)\\
 &=L_X\Omega-\theta(X)\Omega\\
 &=\Omega-\theta(X)\Omega=(1-\theta(X))\Omega,
\end{align*}
Similarly, we have 
\[
d^\theta (X'\lrcorner\Omega) =(1-\theta(X'))\Omega.
\]
Hence, we can give that
\[
d^\theta\left(X'\lrcorner \Omega -\frac{1-\theta(X')}{1-\theta(X)} X\lrcorner \Omega\right) = 0.
\]
Let $\alpha=X'\lrcorner \Omega -\frac{1-\theta(X')}{1-\theta(X)} X\lrcorner \Omega,$ so $d^\theta \alpha=0$ and 
 by using the antisymmetry of the contraction and the definition of Hamiltonian vector fields, we can  compute:
\begin{align*}
X_H\lrcorner\alpha= &X_H\lrcorner\left(X'\lrcorner \Omega -\frac{1-\theta(X')}{1-\theta(X)} X\lrcorner \Omega\right) \\
&= X_H\lrcorner X'\lrcorner \Omega - \frac{1-\theta(X')}{1-\theta(X)}X_H\lrcorner X\lrcorner \Omega \\
&= -X'(H)+H\theta(X') + \frac{1-\theta(X')}{1-\theta(X)}X(H)-\frac{1-\theta(X')}{1-\theta(X)}H\theta(X)\\
&= \left(\frac{1-\theta(X')}{1-\theta(X)}\Lambda - \Lambda'+\theta(X')-\frac{1-\theta(X')}{1-\theta(X)}\theta(X)\right) H.
\end{align*}
This proves the first part.

Conversely, given a $d^\theta$-closed 1-form \( \alpha \) such that \( X_H\lrcorner \alpha = (\Lambda' - \Lambda) H \), define a new vector field \( X' \) via:
\[
X'\lrcorner \Omega =\frac{(1-\theta(X'))X\lrcorner \Omega}{1-\theta (X)}  + \alpha.
\]
Then 
\begin{align*}
d^\theta(X'\lrcorner \Omega) &= d^\theta\left(\frac{(1-\theta(X'))X\lrcorner \Omega}{1-\theta (X)} + \alpha\right)\\
&=\frac{1-\theta(X')}{1-\theta(X)}d^\theta X\lrcorner\Omega\\
&=\frac{1-\theta(X')}{1-\theta(X)}(dX\lrcorner\Omega -\theta\wedge X\lrcorner\Omega)\\
&=\frac{1-\theta(X')}{1-\theta(X)}(L_X\Omega-X\lrcorner d\Omega -\theta\wedge X\lrcorner\Omega)\\
&=\frac{1-\theta(X')}{1-\theta(X)}(L_X\Omega-X\lrcorner (\theta\wedge\Omega)) -\theta\wedge X\lrcorner\Omega)\\
&=\frac{1-\theta(X')}{1-\theta(X)}(L_X\Omega- \theta(X)\Omega) \\
&=\frac{1-\theta(X')}{1-\theta(X)}(\Omega- \theta(X)\Omega) =(1-\theta(X')) \Omega,
\end{align*} so 
\begin{align*}
L_{X'} \Omega&=dX'\lrcorner\Omega+X'\lrcorner d\Omega\\
&=d^\theta X'\lrcorner\Omega+\theta\wedge X'\lrcorner\Omega+X'\lrcorner(\theta\wedge\Omega)\\
&=(1-\theta(X')) \Omega+\theta(X') \Omega=\Omega.
\end{align*}
Moreover,
\begin{align*}
X_H\lrcorner X'\lrcorner \Omega &= X_H\lrcorner\left(\frac{(1-\theta(X'))X\lrcorner \Omega}{1-\theta (X)}  + \alpha\right)\\
&= -\frac{1-\theta(X')}{1-\theta(X)}X(H)+\frac{1-\theta(X')}{1-\theta(X)}\theta(X)H + X_H\lrcorner \alpha\\
& =\left(-\frac{1-\theta(X')}{1-\theta(X)}\Lambda +\frac{1-\theta(X')}{1-\theta(X)}\theta(X) + \Lambda' - \Lambda\right) H ,
\end{align*}
hence $$ X'(H) = X'\lrcorner X_H\lrcorner\Omega  +H\theta(X')=\left(\frac{1-\theta(X')}{1-\theta(X)}\Lambda -\frac{1-\theta(X')}{1-\theta(X)}\theta(X) -\Lambda' +\Lambda+\theta(X')\right)H,$$ showing that \( X' \) is a scaling symmetry of degree \( \frac{1-\theta(X')}{1-\theta(X)}\Lambda -\frac{1-\theta(X')}{1-\theta(X)}\theta(X) -\Lambda' +\Lambda+\theta(X')\).
\end{proof}

 We now introduce a little more general notion of scaling symmetry.
\begin{definition}
A vector filed $X\in\mathfrak{X}(M)$ shall be called a scaling symmetry of degrees $(\Lambda,\beta),$ $\Lambda,\beta\in\mathbb{R}$  for the Hamiltonian system $(M,\Omega,\theta,H)$ when
\begin{enumerate}
\item[i)] $L_{X}\Omega=\beta\Omega$,
\item [ii)] $L_{X}\theta=0$, and
\item [iii)] $L_{X}H=\Lambda H$.
\end{enumerate}
\end{definition}
Using this scaling symmetry, we now establish a result analogous to Lemma \ref{L4}.
\begin{lemma}\label{L5}
If $X$ is a scaling symmetry of degree $(\Lambda,\beta))$  for the Hamiltonian system $(M,\Omega,\theta,H)$ then $L_{X}X_{H}=(\Lambda-\beta)X_{H}$.
\end{lemma}
\begin{proof}
Let assume that  $X$ is a scaling symmetry of degree $(\Lambda,\beta))$ for $(M,\Omega,\theta,H)$, we have
\begin{equation}
\begin{split}
(L_{X}X_{H})\lrcorner\Omega=[X,X_{H}]\lrcorner\Omega&=L_{X}(X_{H}\lrcorner \Omega)-X_{H}\lrcorner L_{X}\Omega\\
&= L_{X}(d^{\theta}H)-\beta X_{H}\lrcorner\Omega\\
&=d^{\theta}(L_{X}H)-HL_{X}\theta-\beta d^{\theta}H\\
&= d^{\theta}(\Lambda H)-\beta d^{\theta}H\\
&= d^{\theta}((\Lambda-\beta)H);
\end{split}
\end{equation}
so $L_{X}X_{H}$ is a Hamiltonian vector field with Hamiltonian function $(\Lambda-\beta)H$, i.e.,
\begin{equation}
L_{X}X_{H}=X_{(\Lambda-\beta)H}=(\Lambda-\beta)X_{H}.
\end{equation}
\end{proof}
\begin{remark}\label{R4}
For any scaling symmetry \( X \) of degree \( (\Lambda, \beta) \) with \( \beta \neq 0 \), the rescaled vector field \( \frac{X}{\beta} \) is a scaling symmetry of degree \( \frac{\Lambda}{\beta} \). Conversely, for any scaling symmetry \( Y \) of degree \( \Lambda \), and any nonzero constant \( \beta \in \mathbb{R} \), the vector field \( \beta Y \) is a scaling symmetry of degree \( (\beta \Lambda, \beta) \).
\end{remark}
\begin{proof}
We compute directly:
\begin{align*}
\mathcal{L}_{\frac{X}{\beta}} \Omega &= \frac{1}{\beta} \mathcal{L}_X \Omega = \frac{1}{\beta} \cdot \beta \Omega = \Omega, \\
\theta\left( \frac{X}{\beta} \right) &= \frac{1}{\beta} \theta(X) = 0, \\
\mathcal{L}_{\frac{X}{\beta}} H &= \frac{1}{\beta} \mathcal{L}_X H = \frac{\Lambda}{\beta} H.
\end{align*}
Hence, \( \frac{X}{\beta} \) is a scaling symmetry of degree \( \frac{\Lambda}{\beta} \).

Conversely, for any scaling symmetry \( Y \) of degree \( \Lambda \), and any constant \( \beta \neq 0 \), we have:
\begin{align*}
\mathcal{L}_{\beta Y} \Omega &= \beta \mathcal{L}_Y \Omega = \beta \Omega, \\
\mathcal{L}_{\beta Y} \theta &= \beta \mathcal{L}_Y \theta = 0, \\
\mathcal{L}_{\beta Y} H &= \beta \mathcal{L}_Y H = \beta \Lambda H.
\end{align*}
Thus, \( \beta Y \) is a scaling symmetry of degree \( (\beta \Lambda, \beta) \).
\end{proof}
Based on Remark \ref{R4}, we know that when \( \beta \neq 0 \), a scaling symmetry \( X \) of degree \( (\Lambda, \beta) \) can be rescaled to obtain a scaling symmetry of degree \( \Lambda \), and conversely. However, such a transformation is no longer valid when \( \beta = 0 \). Nevertheless, scaling symmetries of degree $(\Lambda,0)$ do exist. In the following, we present an example to illustrate this situation.
\begin{example}\label{E3}
Let \( (T^*Q, \Omega, \theta) \) be a LCS  manifold, where
\[
\Omega = dp_i \wedge dq_i - p_i \theta \wedge dq_i, \qquad Q = \mathbb{R}^2 , \quad \theta = q_1\,dq_1.
\]
Consider the Hamiltonian \( H = e^{\Lambda q_2} \), where \( \Lambda \in \mathbb{R} \), and the vector field \( X = \frac{\partial}{\partial q_2} \). Then the corresponding LCS form becomes
\[
\Omega = dp_1 \wedge dq_1 + dp_2 \wedge dq_2 - p_2 q_1\,dq_1 \wedge dq_2,
\]
which is clearly nondegenerate. Moreover, we compute:
\[
L_X \Omega = 0, \qquad L_X \theta = 0, \qquad L_X H = \Lambda H.
\]
Thus, $X=\frac{\partial}{\partial q_2}$ is a scaling symmetry of degree $(\Lambda,0)$.
\end{example}

Now, let us consider a special case of a scaling symmetry of degree \( (\Lambda, \beta) \), namely the case \( (\Lambda, \beta) = (0, 0) \). This means that the vector field \( X \) satisfies
\[
L_X \Omega = 0, \qquad L_X \theta = 0, \qquad L_X H = 0.
\]
In other words, \( X \) preserves the LCS structure as well as the Hamiltonian function. In the following, we provide an example to demonstrate that such symmetries do indeed exist.
\begin{example}
Let $(M,\Omega,\theta)$ be the LCS manifold given by Example \ref{E3}, let $H=e^{\Lambda q_1}$ and $X=\frac{\partial}{\partial q_2}$, we can see that  \[
L_X \Omega = 0, \qquad L_X \theta = 0, \qquad L_X H = 0.
\]
Thus, $X=\frac{\partial}{\partial q_2}$ is a scaling symmetry of degree $(0,0)$.
\end{example}
Next, we present an interesting property of scaling symmetries of degree \( (\Lambda, \beta) \).
\begin{proposition}
The set of all scaling symmetries of degree \( (\Lambda, \beta) \) forms a Lie algebra. Moreover, the Lie bracket of any two scaling symmetries of degree \( (\Lambda, \beta) \) is a scaling symmetry of degree \( (0, 0) \).
\end{proposition}
\begin{proof}
In fact, it suffices to prove that the Lie bracket of any two scaling symmetries of degree \( (\Lambda, \beta) \) is itself a scaling symmetry of degree \( (0, 0) \). 

Let \( X_1 \) and \( X_2 \) be scaling symmetries of degrees \( (\Lambda_1, \beta_1) \) and \( (\Lambda_2, \beta_2) \), respectively. Using the identity 
\[
L_{[X_1, X_2]} = L_{X_1} L_{X_2} - L_{X_2} L_{X_1},
\]
we compute:
\begin{align*}
L_{[X_1, X_2]} \Omega 
&= L_{X_1} L_{X_2} \Omega - L_{X_2} L_{X_1} \Omega 
= \beta_2 L_{X_1} \Omega - \beta_1 L_{X_2} \Omega 
= \beta_2 \beta_1 \Omega - \beta_1 \beta_2 \Omega = 0, \\
L_{[X_1, X_2]} \theta 
&= L_{X_1} L_{X_2} \theta - L_{X_2} L_{X_1} \theta = 0, \\
L_{[X_1, X_2]} H 
&= L_{X_1} L_{X_2} H - L_{X_2} L_{X_1} H 
= \Lambda_2 L_{X_1} H - \Lambda_1 L_{X_2} H 
= \Lambda_2 \Lambda_1 H - \Lambda_1 \Lambda_2 H = 0.
\end{align*}
    Therefore, \( [X_1, X_2] \) is a scaling symmetry of degree \( (0, 0) \), which completes the proof.
\end{proof}

\begin{remark}
The more general concept we present  has the case that \(\beta\) can be zero. However, when \(\beta \neq 0\), we have shown above that this case can be transformed into the standard scaling symmetries. Nevertheless, for standard scaling symmetries, the condition \( L_X \Omega = 0 \) does not occur.
\end{remark}

\subsection{The time-dependent case}
Now we are interested in considering time-dependent transformations. For this, let $(M,\hat{\Omega},\hat{\theta},H)$, with $H\in C^{\infty}(\mathbb{R}\times M)$, be a time-dependent Hamiltonian system on the LCS manifold $(M,\Omega,\theta)$. We present the following definition as a suitable extension of the notion of canonoid transformation for time-dependent Hamiltonian dynamics on LCS manifolds.
\begin{definition}
\label{decanonoidtime}
A canonoid transformation for $(M,\hat{\Omega},\hat{\theta},H)$ is a diffeomorphism $\Psi$ on $\mathbb{R}\times M$ preserving the temporal parameter $t$, for which there exists a function $K\in C^{\infty}(\mathbb{R}\times M)$ such that
\begin{equation}
X_{H}\lrcorner \Psi^{*}\hat{\Omega}=d^{\Psi^{*}\hat{\theta}}K-\frac{\partial K}{\partial t}dt.
\end{equation}
\end{definition}
We can observe that a diffeomorphism $\Psi$ on $\mathbb{R}\times M$, preserving the temporal parameter $t$, is a canonoid transformation for the Hamiltonian $H$ if an only if for each fixed value of $t$, $\Psi_{t}:M\longrightarrow M$ is a conoid transformation for the Hamiltonian $H_{t}\in C^{\infty}(M)$; indeed, for each fixed value of $t$ we have
\begin{equation}
X_{H}\lrcorner \Psi_{t}^{*}\Omega=d^{\Psi_{t}^{*}\theta}K_{t},
\end{equation}
if and only if $\Psi$ is a canonoid transformation.

Of course, by considering (local) one-parameter groups of time-dependent canonoid transformations, we have an analogous result to theorem \ref{theoremcanonoid}, it reads
\begin{theorem}
\label{theoremcanonoidtime}
$X\in\mathfrak{X}(\mathbb{R}\times M)$ is the infinitesimal generator of a one-parameter group of canonoid transformations for $(\mathbb{R}\times M,\hat{\Omega},\hat{\theta},H)$ if and only if $[X,X_{H}]$ is a Hamiltonian vector field.
\end{theorem}

We can also extend the concept of scaling symmetry to the time-dependent case. We consider the following definition.
\begin{definition}
A vector filed $X\in\mathfrak{X}(\mathbb{R}\times M)$ shall be called a scaling symmetry of degree $\Lambda\in\mathbb{R}$  for the Hamiltonian system $(\mathbb{R}\times M,\hat{\Omega},\hat{\theta},H)$ when
\begin{enumerate}
\item [i)] $X\lrcorner dt=0$ (therefore $L_{X}dt=0$)
\item[ii)] $L_{X}\hat{\Omega}=\hat{\Omega}$,
\item [iii)] $L_{X}\hat{\theta}=0$, and
\item [iv)] $L_{X}H=\Lambda H$.
\end{enumerate}
\end{definition}
As in the time-independent framework, scaling symmetries are nonstandard symmetries that rescale the dynamics by a constant factor. Indeed, we have that if $X$ is a scaling symmetry of degree $\Lambda\in\mathbb{R}$  for the Hamiltonian system $(\mathbb{R}\times M,\hat{\Omega},\hat{\theta},H)$ then $L_{X}X_{H}=(\Lambda-1)X_{H}$, which implies that $X$ is the infinitesimal generator of a one-parameter group of noncanonical canonoid transformations.

Here, we can also introduce a little more general notion of scaling symmetry of degree $(\Lambda,\beta)\in\mathbb R^2$ for time-dependent Hamiltonian systems.

\begin{definition}
A vector filed $X\in\mathfrak{X}(\mathbb{R}\times M)$ shall be called a scaling symmetry of degree $(\Lambda,\beta)\in\mathbb{R}^2$  for the Hamiltonian system $(\mathbb{R}\times M,\hat{\Omega},\hat{\theta},H)$ when
\begin{enumerate}
\item [i)] $X\lrcorner dt=0$ (therefore $L_{X}dt=0$)
\item[ii)] $L_{X}\hat{\Omega}=\beta\hat{\Omega}$,
\item [iii)] $L_{X}\hat{\theta}=0$, and
\item [iv)] $L_{X}H=\Lambda H$.
\end{enumerate}
\end{definition}
Similar to the proof of Lemma \ref{L5} for the time-independent Hamiltonian system, we can obtain that if $X$ is a scaling symmetry of degree $(\Lambda, \beta) \in \mathbb{R}^2$ for the Hamiltonian system $(\mathbb{R} \times M, \hat{\Omega}, \hat{\theta}, H)$, then
\[
\mathcal{L}_X X_H = (\Lambda - \beta) X_H.
\]

\section{Symmetries and dissipated quantities}

Given a LCS manifold $(M,\Omega,\theta)$, the LCS structure defines a Jacobi structure on $M$, i.e., a Lie bracket on $C^{\infty}(M)$, called Jacobi bracket, that satisfies the Jacobi identity and the weak Leibniz rule, this last is
$supp(\lbrace f,g\rbrace)\subseteq supp(f)\cap supp(g)$ for $f,g\in C^{\infty}(M)$ \cite{Marle,Ibort1997}. Namely, the Jacobi bracket is defined by
\begin{equation}
\lbrace f,g\rbrace=\Omega(X_{f},X_{g})=X_{g}f-f\theta(X_{g}).
\end{equation}

In terms of the Jacobi bracket, a function $f\in C^{\infty}(M)$ is a constant of motion of the Hamiltonian system $(M,\Omega,\theta,H)$ if
\begin{equation}
\lbrace f,H\rbrace+f\theta(X_{H})=0.
\end{equation}
It is clear that the Hamiltonian function is not a constant of motion, indeed, $\lbrace H,H\rbrace+H\theta(X_{H})=H\theta(X_{H})$, following the language introduced in \cite{BG2023,GGMRR2020,LL2020}, we can say that $H$ is a dissipated quantity.

\subsection{Dissipated quantities and Noether symmetries}
Formally we state the following definition.
\begin{definition}
\(f\in C^{\infty}(M)\) shall be called a dissipated quantity of the Hamiltonian system $(M,\Omega,\theta,H)$  if \(\{H, f\} = 0\). Equivalently, $X_H(f)=f\theta(X_H)$.
\end{definition} 
We note that the set of dissipated functions is a Lie subalgebra of \((C^{\infty}(M), \{\cdot, \cdot\})\). Indeed, \(\mathbb{R}\)-linear combinations of dissipated functions are dissipated, and, because of the Jacobi identity, the Jacobi bracket of two dissipated functions is dissipated. Moreover, it is an algebra over the set of conserved quantities; that is, if \(f\) is a dissipated quantity and \(g\) is a conserved quantity, then \(fg\) is dissipated:
\[X_H(fg) = gX_H(f) = \theta(X_H)fg.\]

If we assume that \(H\) has no zeros, we can relate dissipated functions to conserved functions. Assume that \(f\) is dissipated, then \(f/H\) is a conserved quantity. Indeed:

\[X_H\left(\frac{f}{H}\right) = \frac{X_H(f)H - fX_H(H)}{H^2} = \frac{\theta(X_H)fH - \theta(X_H)fH}{H^2} = 0.\]

In general, if \(f_1, f_2\) commutes with \(H\), then \(f_1/f_2\) is a conserved quantity, assuming \(f_2\) has no zeros.

\begin{remark}
In the particular case where \( \theta(X_H) = 0 \), then the dissipated quantities are precisely the conserved quantities. That is, \(\{H, f\} = 0\) if and only if \( X_H(f) = 0 \).
\end{remark}

Given a dissipated quantity $f$, we have that its Hamiltonian vector field $X_{f}$ is such that
\begin{equation}
L_{X_{f}}X_{H}=[X_{f},X_{H}]=X_{\lbrace f,H\rbrace}=0.
\end{equation}
A vector field $X$ satisfying $L_{X}X_{H}=0$ (its flow preserves the dynamics) is called a dynamical symmetry \cite{prince1983toward,GGMRR2020}. Of course dynamical symmetries have been shown to be infinitesimal generators of one parameter groups of canonoid transformations \cite{azuaje2024scaling}. In the case when a dynamical symmetry is a Hamiltonian vector field it is called a Noether symmetry \cite{Kosmann2011,roman2020summary,jovanovic2016noether,BG2023}. 

Within the framework of symplectic or Poisson geometry, it is well known that Noether symmetries are infinitesimal generators of one-parameter groups of canonical invariance transformations (see \cite{azuaje2025canonical} for a geometric modern review). In our context it is not true in general, indeed, for a dissipated quantity $f$ we have $L_{X_{f}}\Omega=\theta(X_{f})\Omega$ and $X_{f}H=H\theta(X_{f})$, which are zero if and only if $\theta(X_{f})=0$. We conclude the following:
\begin{remark}
If $f$ is a dissipated quantity for the Hamiltonian system $(M,\Omega,\theta,H)$ such that $\theta(X_{f})=0$, then $X_{f}$ is the infinitesimal generator of a one-parameter group of canonical invariance transformations (canonical transformations that leave the Hamiltonian invariant). Observe that it is not required that $f$ is a constant of motion (it is an important difference with the symplectic or Poisson frameworks).
\end{remark}
\begin{example}
Consider \( M = \mathbb{R}^4 \setminus \{0\} \), with 
\[
\Omega = e^x (dx \wedge dy + dw \wedge dz), \quad \theta = dx.
\]
Then \( (M, \Omega, \theta) \) is a LCS manifold. Let 
\[
H = z + \frac{y}{w}, \quad f = w.
\]
We compute the associated Hamiltonian vector fields:
\begin{align*}
    X_H &= e^{-x} \left( \frac{1}{w} \frac{\partial}{\partial x} + \left( z + \frac{y}{w} \right) \frac{\partial}{\partial y} + \frac{\partial}{\partial w} + \frac{y}{w^2} \frac{\partial}{\partial z} \right), \\
    X_f &= e^{-x} \left( w \frac{\partial}{\partial y} - \frac{\partial}{\partial z} \right).
\end{align*}
We also compute the values of \(\theta\) on these vector fields:
\begin{align*}
    \theta(X_f) &= \left( e^{-x} \left( w \frac{\partial}{\partial y} - \frac{\partial}{\partial z} \right) \right) \lrcorner\, dx = 0, \\
    \theta(X_H) &= \left( e^{-x} \left( \frac{1}{w} \frac{\partial}{\partial x} + \left( z + \frac{y}{w} \right) \frac{\partial}{\partial y} + \frac{\partial}{\partial w} + \frac{y}{w^2} \frac{\partial}{\partial z} \right) \right) \lrcorner\, dx = \frac{e^{-x}}{w} \neq 0.
\end{align*}
Now compute the bracket:
\begin{align*}
    \{H, f\} &= X_f(H) - H \theta(X_f) \\
    &= \left( e^{-x} \left( w \frac{\partial}{\partial y} - \frac{\partial}{\partial z} \right) \right) \left( z + \frac{y}{w} \right) = 0.
\end{align*}
Hence, we conclude that \( f \) is a dissipated quantity for the Hamiltonian system \( (M, \Omega, \theta, H) \), satisfying \( \theta(X_f) = 0 \). This implies that \( X_f \) is the infinitesimal generator of a one-parameter group of canonical invariance transformations.
\end{example}

\subsection{The time-dependent case}
We can also define a Jacobi bracket on the extended phase space $(\mathbb{R}\times M,\hat{\Omega},\hat{\theta})$, indeed, given $f,g\in C^{\infty}(\mathbb{R}\times M)$, the Jacobi bracket is 
\begin{equation}
\lbrace f,g\rbrace=\hat{\Omega}(X_{f},X_{g})=X_{g}f-f\hat{\theta}(X_{g}).
\end{equation}

Let $H\in C^{\infty}(\mathbb{R}\times M)$. A function $f\in C^{\infty}(\mathbb{R}\times M)$ is a constant of motion of the Hamiltonian system $(\mathbb{R}\times M,\hat{\Omega},\hat{\theta},H)$ if
\begin{equation}
\lbrace f,H\rbrace+f\hat{\theta}(X_{H})+\frac{\partial f}{\partial t}=0.
\end{equation}
Again, it is clear that the Hamiltonian function is not a constant of motion, indeed, $\lbrace H,H\rbrace+H\theta(X_{H})+\frac{\partial H}{\partial t}=H\theta(X_{H})+\frac{\partial H}{\partial t}$. We could think that $H$ is a dissipated quantity, however it is convenient to consider $H$ as a dissipated quantity only when it is time-independent (as in the symplectic case a time-dependent Hamiltonian is not a constant of motion). We propose the following definition (see \cite{gaset2023symmetries,azuaje2024scaling} for an analogous definition for time-dependent contact mechanics).

\begin{definition}
\(f\in C^{\infty}(\mathbb{R}\times M)\) shall be called a dissipated quantity of the Hamiltonian system $(\mathbb{R}\times M,\hat{\Omega},\hat{\theta},H)$  if \(\{f, H\}+\frac{\partial f}{\partial t}=0\). Equivalently, $X_H(f)=f\theta(X_H)-\frac{\partial f}{\partial t}$.
\end{definition} 
Whenever $H$ depends explicitly on $t$ it is not itself a dissipated quantity.

As in the time-independent case, the quotient of dissipated quantities (where defined) are constants of motion.

Given a dissipated quantity $f$, we have that its Hamiltonian vector field $X_{f}$ is a dynamical symmetry (its flow preserves the dynamical vector field $\tilde{X}_{H}$), i.e.,
\begin{equation}
L_{X_{f}}\tilde{X}_{H}=0.
\end{equation}
Indeed, first let us see that for each function $g\in C^{\infty}(\mathbb{R}\times M)$ we have
\begin{equation}
\begin{split}
[\frac{\partial}{\partial t},X_{g}]\lrcorner dt= L_{\frac{\partial}{\partial t}}(X_{g}\lrcorner dt)-X_{g}\lrcorner L_{\frac{\partial}{\partial t}}dt=0
\end{split}
\end{equation}
and
\begin{equation}
\begin{split}
[\frac{\partial}{\partial t},X_{g}]\lrcorner \hat{\Omega}&= L_{\frac{\partial}{\partial t}}(X_{g}\lrcorner\hat{\Omega})-X_{g}\lrcorner L_{\frac{\partial}{\partial t}}\hat{\Omega}\\
&=L_{\frac{\partial}{\partial t}}(d^{\hat{\theta}}g-\frac{\partial g}{\partial t}dt)\\
&=d^{\hat{\theta}}L_{\frac{\partial}{\partial t}}g-gL_{\frac{\partial}{\partial t}}\hat{\theta}-(L_{\frac{\partial}{\partial t}}\frac{\partial g}{\partial t})dt-\frac{\partial g}{\partial t}L_{\frac{\partial}{\partial t}}dt\\
&=d^{\hat{\theta}}\frac{\partial g}{\partial t}-\frac{\partial}{\partial t}(\frac{\partial g}{\partial t})dt,
\end{split}
\end{equation}
i.e., $[\frac{\partial}{\partial t},X_{g}]$ is a Hamiltonian vector field with Hamiltonian function $\frac{\partial g}{\partial t}$. So we have
\begin{equation}
L_{X_{f}}\tilde{X}_{H}=[X_{f},\tilde{X}_{H}]=[X_{f},X_{H}]+[X_{f},\frac{\partial}{\partial t}]=X_{\lbrace H,f\rbrace}+X_{-\frac{\partial g}{\partial t}}=X_{\lbrace H,g\rbrace-\frac{\partial g}{\partial t}}.
\end{equation}
For a dissipated quantity $f$ we have
\begin{equation}
L_{X_{f}}\tilde{X}_{H}=X_{\lbrace H,f\rbrace-\frac{\partial f}{\partial t}}=-X_{\lbrace f,H\rbrace+\frac{\partial f}{\partial t}}=X_{0}=0.
\end{equation}

As in the time-independent case, dynamical symmetries are infinitesimal generators of one parameter groups of (time-dependent) canonoid transformations. Following the previous language, when a dynamical symmetry is a Hamiltonian vector field we shall call it a Noether symmetry. The following remark is in order:
\begin{remark}
If $f$ is a dissipated quantity for the Hamiltonian system $(\mathbb{R}\times M,\hat{\Omega},\hat{\theta},H)$ such that $\hat{\theta}(X_{f})=0$, then $X_{f}$ is the infinitesimal generator of a one-parameter group of (time-dependent) canonical invariance transformations.
\end{remark}

\section*{Acknowledgment}
The first author (R. A.) wishes to thank the financial support provided by the Secretaría de Ciencia, Humanidades, Tecnología e Innovación (SECIHTI) of Mexico through a postdoctoral fellowship under the Estancias Posdoctorales por México 2022 program.

The research of the second author (X. Z.) is supported by NSFC (Grant No. 12401234).
	$\\$
	
{\noindent$\mathbf{Conflict\;of\;interest\;statement.}$ On behalf of all authors, the corresponding author states that there is no conflict of interest.
	
	$\\$
\noindent$\mathbf{Data\;availability.}$ Data sharing is not applicable to this article as no new data were created or analyzed in this study.


\begin{thebibliography}{00}

\bibitem{Abraham2} R. Abraham and  J. E. Marsden,  Foundations of Mechanics, Advanced Book Program, (1978).

\bibitem{Asorey1983} Asorey M, Cariñena J F and Ibort L A 1983 Generalized canonical transformations for time-dependent systems J. Math. Phys. 24 2745–50.

\bibitem{azuaje2024scaling} R. Azuaje and A. Bravetti. Scaling symmetries and canonoid transformations in hamiltonian systems. Int. J. Geom. Methods Mod. Phys., 21(04):2450077, 2024.

\bibitem{azuajecanonical2023} R. Azuaje and A. M. Escobar-Ruiz. Canonical and canonoid transformations for Hamiltonian systems on (co)symplectic and (co)contact manifolds. J. Math. Phys., 64(3):033501, 2023.

\bibitem{azuaje2025canonical}  R. Azuaje and A. M. Escobar-Ruiz. Canonical transformations: from the coordinate based approach to the geometric one. Phys. Scr. 100 065228 (2025).

\bibitem{Bande} G. Bande, and D. Kotschick,  Contact pairs and locally conformally symplectic structures. In: Loubeau, E., Montaldo, S. (eds.) Harmonic maps and differential geometry, Contemp. Math., vol. 542, pp. 85–98. American Mathematical Society, Providence (2011).

\bibitem{Belgun} F. Belgun,  On the metric structure of non-Kähler complex surfaces, Math. Ann. 317, 1–40, (2000).

\bibitem{Borman} M. S. Borman, Y. Eliashberg, and E. Murphy, Existence and classification of overtwisted contact structures in all dimensions, Acta Math. 2, 281–361 (2015).

\bibitem{Bourgeois} F. Bourgeois, Odd dimensional tori are contact manifolds, Internat. Math. Res. Not.  30, 115–120, (2002). 

\bibitem{Bowden} J. Bowden, D. Crowley, and A. I. Stipsicz, Contact structures on $M\times S^2$, Math. Ann. 358, 351–359,  (2014). 

\bibitem{BG2023} A. Bravetti and A. Garcia-Chung. A geometric approach to the generalized Noether theorem.
J. Phys. A: Math. Theor., 54, 095205 (2023).

\bibitem{bravetti2024kirillov} A. Bravetti, S. Grillo, J. C. Marrero, and E. Padron. Kirillov structures and reduction of hamiltonian systems by
scaling and standard symmetries. Stud. Appl. Math., 2024.

\bibitem{bravetti2023scaling} A. Bravetti, C. Jackman, and D. Sloan. Scaling symmetries, contact reduction and poincar´e’s dream. J. Phys. A: Math. Theor., 56(43):435203, 2023.

\bibitem{Calkin} Calkin M G 1996 Lagrangian and Hamiltonian Mechanics (World Scientiﬁc Publishing).

\bibitem{Carinena} J. F. Cari$\tilde n$ena, F. Falceto and  M. F.  Ra$\tilde n$ada, Canonoid transformations and master symmetries, J. Geom. Mech. 5, 151–66, (2013).

\bibitem{CR88} J. F. Cari$\tilde{n}$ena and M. F. Ra$\tilde{n}$ada. Canonoid transformations from a geometric perspective. J. Math. Phys., 29:2181–2186, 1988.

\bibitem{cariñena1989} Jose F. Carinena,  and F. Ranada. Manuel, "Poisson maps and canonoid transformations for time‐dependent Hamiltonian systems." Journal of mathematical physics 30.10 (1989): 2258-2266.

\bibitem{cariñena1985canonical} J. F. Cariñena, Gomis J, Ibort L A and Román N 1985 Canonical transformations theory for presymplectic systems J. Math. Phys. 26 1961–9.

\bibitem{Chantraine} B. Chantraine and E. Murphy,  Conformal symplectic geometry of cotangent bundles, J. Symplectic Geom. 17, 639–661, (2019).

\bibitem{Chinea} D. Chinea, M. de Le\'on and J. C. Marrero, Locally conformal cosymplectic manifolds and time-dependent Hamiltonian systems, Comment. Math.Univ.Carolin. 32,  383387, (1991).

\bibitem{Eliashberg} Y. Eliashberg, E. Murphy,  Making cobordisms symplectic,  J. Amer. Math. Soc. 36,  1–29, (2023).

\bibitem{Esen} O. Esen, M. de León, C. Sardón, and M. Zajsc, Hamilton–Jacobi formalism on locally conformally symplectic manifolds, J. Math. Phys. 62, 033506, (2021).

\bibitem{GGMRR2020} J. Gaset, X. Grácia, M.C. Muñoz-Lecanda, X. Rivas, and N. Román-Roy. New contributions
to the Hamiltonian and Lagrangian contact formalisms for dissipative mechanical systems
and their symmetries. Int. J. Geom. Meth. Mod. Phys., 17: 2050090 (2020).

\bibitem{gaset2023symmetries} J. Gaset, A. López-Gordón and X. Rivas, Symmetries, conservation and dissipation
in time-dependent contact systems, Fortschr. Phys. 71 (2023) 2300048.

\bibitem{Goldstein} Goldstein H, Poole C and Safko J 2002 Classical Mechanics 3d ed. (Addison-Wesley).

\bibitem{Gray} A. Gray and  L. M. Hervella, The sixteen classes of almost Hermitian manifolds and their linear invariants, Ann. Mat. Pura Appl. 123, 35–58, (1980).

\bibitem{Guedira} F. Guedira and A. Lichnerowicz, Géométrie des algèbres de Lie locales de Kirillov, J. Math. Pures Appl. 63, 407–484, (1984).

\bibitem{Ibort1997} A. Ibort et al. Reduction of jacobi manifolds. J. Phys. A: Math. Gen., 30(8):2783, 1997.

\bibitem{jovanovic2016noether} B. Jovanovic, Noether symmetries and integrability in time-dependent Hamiltonian
mechanics, Theor. Appl. Mech. 43(2) (2016) 255–273.

\bibitem{Kosmann2011} Y. Kosmann-Schwarzbach and B. E. Schwarzbach, The Noether Theorems, Invariance and Conservation Laws in the Twentieth Century (Springer, New York, 2011).

\bibitem{Landau} Landau L D and Lifshitz E M 1982 Mechanics 1 (Elsevier Science).

\bibitem{Lee} H. C. Lee,  A kind of even-dimensional differential geometry and its application to exterior calculus, Am. J. Math. 65, 433–438, (1943).

\bibitem{Lefebvre} J. Lefebvre,  Propriétés du groupe des transformations conformes et du groupe des automorphismes d’une variété localement conformément symplectique. C. R. Acad. Sci. Paris Sér. A-B. 268, A717–A719, (1969).

\bibitem{Libermann} P. Libermann,  Sur les structures presque complexes et autres structures infinit\'esimales r\'eguli\'eres, Bull. Soc. Math. Fr. 83, 195–224, (1955).

\bibitem{LL2020} M. de León and M. Lainz. Infinitesimal symmetries in contact Hamiltonian systems. J. Geom. Phys., 153, 103651 (2020).

\bibitem{Marle} C. M. Marle. On Jacobi manifolds and Jacobi bundles. In Symplectic Geometry, Groupoids,
and Integrable Systems, 227–246, New York, NY, 1991. Springer US.

\bibitem{Pardon} J Pardon, Contact homology and virtual fundamental cycles, J. Amer. Math. Soc. 32,  825–919, (2019).

\bibitem{prince1983toward} G. Prince, Toward a classification of dynamical symmetries in classical mechanics, Bull. Aust. Math. Soc. 27(1) (1983) 53–71.

\bibitem{Rastelli2015} G. Rastelli and M. Santoprete, Canonoid and poissonoid transformations, symmetries
and bi-Hamiltonian structures, J. Geom. Mech. 7 (2015) 483–515.

\bibitem{roman2020summary}N. Román-Roy, A summary on symmetries and conserved quantities of autonomous
Hamiltonian systems, J. Geom. Mech. 12(3) (2020) 541–551.

\bibitem{Saletan} E. J. Saletan and A. H. Cromer, ``Theoretical Mechanics,” John Wiley $\&$ Sons, 1971.

\bibitem{Saletan2} E. J. Saletan and J. V. Jos\'e, ``Classical Mechanics: A Contemporary Approach," Cambridge Univ. Press, Cambridge, 1998.

\bibitem{spivak2010} M. Spivak. Physics for Mathematicians: Mechanics I. Publish or Perish, 2010.

\bibitem{Struckmeier2005} J. Struckmeier, Hamiltonian dynamics on the symplectic extended phase space for autonomous and non-autonomous systems
J. Phys. A: Math. Gen. 38 1257 (2005).

\bibitem{Vaisman} I. Vaisman,  Locally conformal symplectic manifolds, Int. J. Math. Math. Sci. 8, 521–536, (1985).

\bibitem{Verbitsky} M. Verbitsky, V. Vuletescu, and L. Ornea, Classification of non-Kähler surfaces and locally conformally Kähler geometry, Russian Math. Surveys. 76, 261–290,  (2021).

\bibitem{Wojtkowski} M. P. Wojtkowski and C. Liverani,  Conformally symplectic dynamics and symmetry of the Lyapunov spectrum, Commun. Math. Phys. 194, 47–60, (1998).

\bibitem{ZSR2023} M. Zajac, C. Sardón and O. Ragnisco. Time-Dependent Hamiltonian Mechanics on a Locally Conformal Symplectic Manifold. Symmetry, 15(4), 843 (2023). 

\bibitem{Zhao1} X. F. Zhao, Canonoid transformation and master symmetries of Hamiltonian systems on locally conformal symplectic manifold, J. Math. Phys. 66 (2025).

\end{thebibliography}
\end{document}